\documentclass[times]{oupau2}

\usepackage{amsmath,amsthm,amssymb}
\usepackage{graphicx}
\usepackage[dvipsnames]{xcolor}
\usepackage[colorlinks=true,citecolor=black,linkcolor=black,urlcolor=blue]{hyperref}
\usepackage{chessboard}
\usepackage{skak}

\newcommand{\arxiv}[1]{\href{http://arxiv.org/abs/#1}{\texttt{arXiv:#1}}}
\definecolor{brown(web)}{rgb}{0.65, 0.16, 0.16}

\newcommand{\cola}[1]{\textcolor{black}{#1}}
\newcommand{\colb}[1]{\textcolor{green}{#1}}
\newcommand{\colc}[1]{\textcolor{red}{#1}}
\newcommand{\cold}[1]{\textcolor{blue}{#1}}
\newcommand{\cole}[1]{\textcolor{Cyan}{#1}}
\newcommand{\colf}[1]{\textcolor{Plum}{#1}}
\newcommand{\colg}[1]{\textcolor{teal}{#1}}

\renewcommand\mod{\,\operatorname{mod}}
\renewcommand\ell{l}
\begin{document}

% Enter full title and short title for running headers
\title{Closed-Form Expressions for the n-Queens Problem and Related Problems}
%\shorttitle{A Demonstration of the \textit{OUP Journal} Class File}

% Enter the publication year and the ID number of the paper
\volumeyear{2016}
\paperID{rnn999}

% Author name(s)
\author{Kevin Pratt}
% Abbreviated author name for running headers
\abbrevauthor{K. Pratt}
% Abbreviated author name for first page header
\headabbrevauthor{Pratt, K.}

\address{%
University of Connecticut, 324 Whitney Hall, Storrs, CT 06269\\
\href{mailto:kevin.pratt@uconn.edu}{kevin.pratt@uconn.edu}}

%\correspdetails{kevin.pratt@uconn.edu}

% Received/revised/accepted dates will be entered by the publisher during production of an accepted paper. Please do not edit these placeholders for submission.
\received{1 Month 20XX}
\revised{11 Month 20XX}
\accepted{21 Month 20XX}

% Enter details of editor communicating this article
\communicated{A. Editor}

\begin{abstract}
In this paper, we derive simple closed-form expressions for the $n$-queens problem
and three related problems in terms of permanents of $(0,1)$ matrices. These formulas are the first of their kind.
Moreover, they provide the first method for solving these problems with polynomial space that has a
nontrivial time complexity bound. We then show how a closed form for the number of Latin squares of order $n$ 
follows from our method. Finally, we prove lower bounds. In particular, we show that the permanent of Schur's complex-valued matrix 
is a lower bound for the toroidal semi-queens problem, or equivalently, the number of transversals in a cyclic Latin square.
\end{abstract}

\maketitle

\section{Introduction}

The $n$-queens problem is to determine $Q(n)$, the number of arrangements of $n$ queens on an $n$-by-$n$ chessboard such that no two queens attack. It is a generalization of the eight queens puzzle posed in 1848 by Max Bezzel, a German chess player. The $n$-queens problem has been widely studied since then, attracting the attention of P\'olya and Lucas. It is now best known as a toy problem in algorithm design \cite{survey}.

Despite this rich history, little is known of the general behavior of $Q(n)$. Key results are that $Q(n) > 1$ for $n > 3$, and $Q(n) > 4^{n/5}$ when $\gcd(n,30) = 5$. See \cite{survey} for a comprehensive survey. The only closed-form expression\footnote{We would like to correct a misunderstanding in \cite{survey}. The authors state that there exists no closed-form expression for $Q(n)$ because it was shown to be beyond the $\#P$ complexity class. However, the result referenced only shows that the $n$-queens problem is beyond $\#P$ because $Q(n)$ can be more than polynomial in $n$ \cite{hard}. A function can clearly be beyond $\#P$ for this reason and still have a closed-form expression; consider $2^n$ for instance.} we are aware of was given in \cite{closed1}. It is ``very complicated" in the authors' own words, however.

The variants of the $n$-queens problem we consider are the \textit{toroidal} $n$-queens problem $T(n)$, the \textit{semi-queens} problem $S(n)$, and the \textit{toroidal semi-queens} problem $TS(n)$. As with $Q(n)$, the general behavior of these functions is not well understood; asymptotic lower bounds are only known for $TS(n)$ \cite{asymp}. 

In this paper, we derive closed-form expressions for $Q(n), T(n), S(n)$, and $TS(n)$ in terms of permanents of $(0,1)$ matrices. The method we use is general and proceeds as follows. First, we come up with an \textit{obstruction matrix} for a problem. Each entry in this matrix is a multilinear monomial. We then prove a formula for the sum of the coefficients of the terms containing some number of distinct variables in a polynomial. This is then used to obtain closed-form expressions for our problems. The expressions we obtain are very similar to those for the number of Latin squares of order $n$, such as those given in \cite{latin}. In fact, we show that one such formula is an immediate corollary of our method.

The permanent was previously considered by Rivin and Zabih to compute $Q(n)$ and $T(n)$ \cite{nqp}. Similarly, in 1874 Gunther used the determinant to construct solutions to the $n$-queens problem for small values of $n$ \cite{survey}. As far as we can tell however, no one has previously attempted to obtain closed-form expressions with this approach. The expressions we obtain in doing so can be evaluated in nontrivial time (i.e., better than the $O(n!)$ brute-force approach) and with polynomial space. The only other algorithms for computing $Q(n)$ and $T(n)$ with nontrivial time complexity bounds were given in \cite{alg}; however, this approach requires exponential space. We are not aware of any previously known algorithms for computing $S(n)$ and $TS(n)$ with nontrivial complexity bounds.

Finally, we prove lower bounds for these problems in terms of determinants of $(0,1)$ matrices. As a consequence, we show that the permanent of Schur's complex-valued matrix \cite{schur} provides a lower bound for the toroidal semi-queens problem.

\section{Preliminary Definitions}
The permanent of an $n$-by-$n$ matrix $\mathbf{A}= (a_{i,j})$ is given by
\begin{equation*}
\mathrm{per}(\mathbf{A}) = \sum_{\sigma \in S_n}\prod_{i=1}^{n}a_{i,\sigma(i)}
\end{equation*}
where $S_n$ is the symmetric group on $n$ elements. It is a well-known result in complexity theory that computing the permanent of a matrix is intractable, even when restricted to the set of $(0,1)$ matrices \cite{sharpp}.

An \textit{obstruction matrix} $\mathbf{A}$ is a square matrix whose entries are multilinear monomials. If $\mathbf{A}$ contains the variables $x_1, x_2, \ldots, x_m$ and $s = (s_i) \in \{0,1\}^m$, then $\mathbf{A}|s$ is the matrix obtained by substituting $x_i = s_i$ for all $i$.

An $n$-by-$n$ matrix $\mathbf{M} = (m_{i,j})$ is \textit{diagonally constant} if each northwest-southeast diagonal is constant; that is, $m_{i,j} = m_{i+1,j+1}$. A \textit{circulant matrix} is a diagonally constant matrix with the property that each row is obtained by rotating the preceding row one position to the right, i.e., $m_{i,j} = m_{i+1, j+1 \mod n}$.

$Q(n)$ is the number of arrangements of $n$ queens on an $n$-by-$n$ chessboard such that no two attack; that is, lie on the same row, column, or diagonal \cite[Sequence A000170] {oeis}.

$S(n)$ is the number of arrangements of $n$ nonattacking \textit{semi-queens} on an $n$-by-$n$ chessboard \cite[Sequence A099152]{oeis}. A semi-queen has the same moves as a queen except for the northeast-southwest diagonal moves. Note that $S(n) \ge Q(n)$.

$T(n)$ is the number of arrangements of $n$ nonattacking queens on a toroidal $n$-by-$n$ chessboard \cite[Sequence A051906]{oeis}. The toroidal board is obtained by identifying the edges of the board as if it were a torus. As a result, the diagonals a queen can move along wrap around the board. Note that $Q(n) \ge T(n)$.

$TS(n)$ is the number of arrangements of $n$ nonattacking semi-queens on an $n$-by-$n$ toroidal chessboard. $TS(n)$ is also the number of transversals in a cyclic Latin square \cite[Sequence A006717]{oeis}. Note that $S(n) \ge TS(n)$.

\begin{figure}
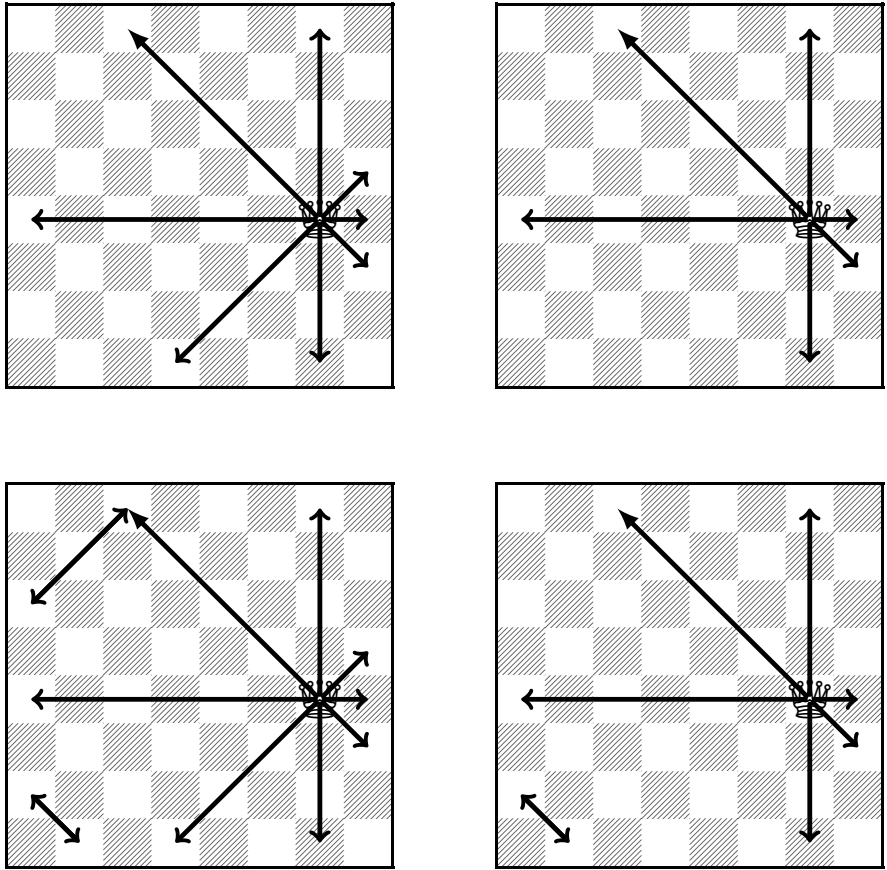

\qquad
\qquad
\qquad
\setchessboard{showmover=false, boardfontsize=18pt, label=false}
\newgame
\chessboard[clearboard,
addpieces={Qg4},
pgfstyle=straightmove,
markmove=g4-c8,
arrow=to,linewidth=0.2ex,
markmove = g4-h3,
markmove = g4-h5,
markmove = g4-d1,
markmove = g4-g8,
markmove = g4-g1,
markmove = g4-h4,
markmove = g4-a4]
\setchessboard{showmover=false, boardfontsize=18pt}
\newgame
\chessboard[clearboard,
addpieces={Qg4},
pgfstyle=straightmove,
markmove=g4-c8,
arrow=to,linewidth=0.2ex,
markmove = g4-h3,
markmove = g4-g8,
markmove = g4-g1,
markmove = g4-h4,
markmove = g4-a4]

\qquad 
\qquad
\qquad
\setchessboard{showmover=false, boardfontsize=18pt}
\newgame
\chessboard[clearboard,
addpieces={Qg4},
pgfstyle=straightmove,
markmove=g4-c8,
arrow=to,linewidth=0.2ex,
markmove = g4-h3,
markmove = b1-a2,
markmove = a2-b1,
markmove = g4-h5,
markmove = g4-d1,
markmove = a6-c8,
markmove = c8-a6,
markmove = g4-g8,
markmove = g4-g1,
markmove = g4-h4,
markmove = g4-a4]
\setchessboard{showmover=false, boardfontsize=18pt}
\newgame
\chessboard[clearboard,
addpieces={Qg4},
pgfstyle=straightmove,
markmove=g4-c8,
arrow=to,linewidth=0.2ex,
markmove = g4-h3,
markmove = b1-a2,
markmove = a2-b1,
markmove = g4-g8,
markmove = g4-g1,
markmove = g4-h4,
markmove = g4-a4]
\caption{From top left to bottom right: The squares attacked by a queen, a semi-queen, a toroidal queen, and a toroidal semi-queen.}
\end{figure}

%%%%%%%%%%%%%%%%%%%%%%%%%%%%%%%%%%%%%%%%%%%%%%%%%%%%%%%
\section{Derivation of the Main Results}
We begin by introducing the $n$-by-$n$ obstruction matrices $\mathbf{Q}_n$, $\mathbf{T}_n$, $\mathbf{S}_n$, and $\mathbf{Z}_n$, which will be used to compute $Q(n)$, $T(n)$, $S(n)$, and $TS(n)$, respectively.

$\mathbf{Q}_n$ contains the variables $x_1, y_1, \ldots, x_{2n-1}, y_{2n-1}$. The variable $x_i$ corresponds to the $i$th northwest-southeast diagonal (indexed from bottom left to top right), and $y_i$ corresponds to the $i$th northeast-southwest diagonal (indexed from bottom right to top left). The $(i,j)^{\text{th}}$ entry of $\mathbf{Q}_n$ is $x_{n-i+j}y_{2n-i-j+1}$. 

$\mathbf{T}_n$ contains the variables $x_1, y_1, \ldots, x_{2n}, y_{2n}$. The variable $x_i$ corresponds to the $i$th northwest-southeast broken diagonal, and $y_i$ corresponds to the $i$th northeast-southwest broken diagonal. The $(i,j)^{\text{th}}$ entry of $\mathbf{T}_n$ is $x_{(n-i+j) \mod n} y_{(2n-i-j+1) \mod n}$.

$\mathbf{S}_n$ contains the variables $x_1, x_2, \ldots, x_{2n-1}$, and $x_i$ corresponds to the $i$th northwest-southeast diagonal. The $(i,j)^{\text{th}}$ entry of $\mathbf{S}_n$ is $x_{n-i+j}$.

$\mathbf{Z}_n$ contains the variables $x_1, x_2, \ldots, x_n$, and $x_i$ corresponds to the $i$th northwest-southeast broken diagonal. The $(i,j)^{\text{th}}$ entry of $\mathbf{Z}_n$ is $x_{(n-i+j)\mod n}$.

\begin{example}
\centering
Obstruction matrices for $Q(n),T(n),S(n)$, and $TS(n)$.
\begin{align*}
&\mathbf{Q}_4 = \begin{bmatrix} 
\cold{x_4}\colg{y_7} & \cole{x_5}\colf{y_6} & \colf{x_6}\cole{y_5} & \colg{x_7}\cold{y_4}\\
\colc{x_3}\colf{y_6} & \cold{x_4}\cole{y_5} & \cole{x_5}\cold{y_4} & \colf{x_6}\colc{y_3}\\
\colb{x_2}\cole{y_5} & \colc{x_3}\cold{y_4} & \cold{x_4}\colc{y_3} & \cole{x_5}\colb{y_2}\\
\cola{x_1}\cold{y_4} & \colb{x_2}\colc{y_3} & \colc{x_3}\colb{y_2} & \cold{x_4}\cola{y_1}
\end{bmatrix}
&&\mathbf{T}_4 = \begin{bmatrix} 
\cold{x_4}\cola{y_1} & \cola{x_1}\colb{y_2} & \colb{x_2}\colc{y_3} & \colc{x_3}\cold{y_4}\\
\colc{x_3}\colb{y_2} & \cold{x_4}\colc{y_3} & \cola{x_1}\cold{y_4} & \colb{x_2}\cola{y_1}\\
\colb{x_2}\colc{y_3} & \colc{x_3}\cold{y_4} & \cold{x_4}\cola{y_1} & \cola{x_1}\colb{y_2}\\
\cola{x_1}\cold{y_4} & \colb{x_2}\cola{y_1} & \colc{x_3}\colb{y_2} & \cold{x_4}\colc{y_3}
\end{bmatrix}\\
&\mathbf{S}_4 = \begin{bmatrix} 
\cold{x_4} & \cole{x_5} & \colf{x_6} & \colg{x_7}\\
\colc{x_3} & \cold{x_4} & \cole{x_5} &\colf{x_6}\\
\colb{x_2} &\colc{x_3} & \cold{x_4} & \cole{x_5}\\
\cola{x_1} & \colb{x_2} & \colc{x_3} & \cold{x_4}
\end{bmatrix}
&&\mathbf{Z}_4 = \begin{bmatrix}
\cold{x_4} & \cola{x_1} & \colb{x_2} & \colc{x_3}\\
\colc{x_3} & \cold{x_4} & \cola{x_1} & \colb{x_2}\\
\colb{x_2} &\colc{x_3} & \cold{x_4} & \cola{x_1}\\
\cola{x_1} & \colb{x_2} & \colc{x_3} & \cold{x_4}
\end{bmatrix}
\end{align*}
\end{example}

\begin{definition}
Let $P$ be a polynomial, and let $k \in \mathbb{N}$. Then $g(P, k)$ is defined to be the sum of the coefficients of the terms in $P$ that are a product of exactly $k$ distinct variables. 
\end{definition}

Note that when $k = \deg{P}$, the terms whose coefficients are summed by $g(P,k)$ are multilinear. This leads to the following fact:

\begin{lemma}
\label{grel}
$g(\mathrm{per}(\mathbf{Q}_n), 2n) = Q(n)$, $g(\mathrm{per}(\mathbf{T}_n), 2n) = T(n)$, $g(\mathrm{per}(\mathbf{S}_n), n) = S(n)$, and $g(\mathrm{per}(\mathbf{Z}_n), n) = TS(n)$.
\end{lemma}

\begin{proof}
This follows immediately from the definition of the permanent and the structure of $\mathbf{Q}_n, \mathbf{T}_n, \mathbf{S}_n,$ and $\mathbf{Z}_n$. Consider $\mathrm{per}(\mathbf{Q}_n)$ for instance. We can write this as a sum of $n!$ terms of degree $2n$. Each term in this polynomial corresponds to a permutation matrix. If a term is square-free, then from the definition of $\mathbf{Q}_n$ no two elements in the corresponding permutation matrix lie along the same diagonal. Since a permutation matrix has no two nonzero entries on the same row or column, it follows that this permutation matrix corresponds to a solution for the $n$-queens problem.
\end{proof}

Suppose that $P$ is a polynomial in $m$ variables.
Let $S_{m,k}$ be the subset of $\{0,1\}^m$ that consists of the tuples containing $k$ ones; that is,
$$S_{m,k} = \{(s_1, \ldots , s_{m}) \in \{0,1\}^{m} : \sum_{i=1}^{m} s_i = k\}.$$
Define $$f(P, k) = {\sum_{(s_1, \ldots , s_m) \in S_{m,k}}} P(s_1, \ldots, s_m).$$
The following fact is now used to derive an expression for $g$ in terms of $f$.
\begin{fact}
Let $m \ge k \ge l \ge 0$. Assume
\begin{align*}
a(k) &= \sum_{i=l}^k b(i) \binom{m-i}{k-i}. \\ \intertext{Then}
b(k) &= \sum_{i=l}^k a(i) \binom{m-i}{k-i} (-1)^{k-i}.
\end{align*}
\end{fact}

\begin{theorem}
\label{cf}
Let $P$ be a polynomial in $m$ variables, and let $1 \le k \le m$. Then
\begin{equation}
g(P, k) = \sum_{i=1}^k (-1)^{i+k} f(P, i) \binom{m-i}{k-i}.
\end{equation}
\end{theorem}

\begin{proof}
Consider a term in $P$ that is a product of $i$ distinct variables where $i \le k$. It follows from the definition of $f$ 
that the coefficient of this term is counted by $f(P,k)$ a total of $\binom{m-i}{k-i}$ times. Therefore
$$f(P, k) = \sum_{i=1}^{k} g(P, i) \binom{m-i}{k-i}.$$
Then by applying Fact 3.4 with $a(k) = f(P,k)$, $b(k) = g(P,k)$, and $l = 1$, equation (1) follows.
\end{proof}

The following expressions follow directly from Lemma \ref{grel} and Theorem \ref{cf}, and the fact that $\mathrm{per}(\mathbf{Q}_n),$ $\mathrm{per}(\mathbf{T}_n),$ $\mathrm{per}(\mathbf{S}_n),$ and $\mathrm{per}(\mathbf{Z}_n)$ are polynomials in $4n-2$, $2n$, $2n-1$, and $n$ variables, respectively.

\begin{theorem}
Let $S_{m,k}$ be the subset of $\{0,1\}^m$  that consists of the tuples containing $k$ ones, $U_n$ the set of all $n$-by-$n$ $(0,1)$ diagonally constant matrices, and $V_n$ the set of all $n$-by-$n$ $(0,1)$ circulant matrices. Then the following identities hold:
\begin{align*}
Q(n) &= \sum_{i=1}^{2n} (-1)^{i} \binom{4n-i-2}{2n-i} \sum_{s \in S_{4n-2, i}} \mathrm{per}(\mathbf{Q}_n |s), \\
T(n) &= \sum_{i=1}^{2n} (-1)^{i+n} \sum_{s \in S_{2n, i}} \mathrm{per}(\mathbf{T}_n |s), \\
S(n) &= \sum_{\mathbf{M} \in U_n} (-1)^{\gamma(\mathbf{M})+n} \mathrm{per}(\mathbf{M}) \binom{2n - \gamma(\mathbf{M}) -1}{n-\gamma(\mathbf{M})},\\
TS(n) &= \sum_{\mathbf{M} \in V_n} (-1)^{\sigma(\mathbf{M})+n} \mathrm{per}(\mathbf{M}),
\end{align*}
where $\gamma(\mathbf{M})$ is the number of nonzero diagonals in $\mathbf{M}$, and $\sigma(\mathbf{M})$ is the number of ones in the first row of $\mathbf{M}$.
\end{theorem}
Note that multiple $(0,1)$ variable assignments to $\mathbf{Q}_n$ and $\mathbf{T}_n$ can correspond to the same $(0,1)$ matrix. As a result, one can think of the formulas for $Q(n)$ and $T(n)$ as summing over multisets of $(0,1)$ matrices. In the cases of $\mathbf{S}_n$ and $\mathbf{Z}_n$, there is a one-to-one relationship between $(0,1)$ variable assignments and $(0,1)$ matrices, so we can write $S(n)$ and $TS(n)$ as sums over sets of $(0,1)$ matrices.

%%%%%%%%%%%%%%%%%%%%%%%%%%%%%%%%%%%%%%%%%%%%%%%%%%%%%%%
\subsection{Complexity Analysis}
The above expressions are impractical to evaluate even for small values of $n$; however, they do provide nontrivial time complexity bounds.

\begin{corollary}
$Q(n)$, $T(n)$, $S(n)$, and $TS(n)$ can be computed in quadratic space and in time $O(n 32^n)$, $O(n 8^n)$, $O(n 8^n),$ and $O(n 4^n)$, respectively. 
\end{corollary}
\begin{proof}
We can compute $Q(n)$ as follows. There are $O(2^{4n})$ $(0,1)$-tuples to enumerate in the summation. For each such tuple $s$, we compute $\mathbf{Q}_n | s$ in $O(n^2)$ time and space, and compute the permanent of this matrix in $O(n2^n)$ time and with $O(n^2)$ space using Ryser's formula \cite{permc}, which states that
$$\mathrm{per}(\mathbf{A}) = \sum_{S\subseteq\{1,\dots,n\}} (-1)^{|S| + n} \prod_{i=1}^n \sum_{j\in S} a_{ij}.$$
%https://www.amazon.com/Permanents-Encyclopedia-Mathematics-its-Applications/dp/0521175143#reader_0521175143
Thus $Q(n)$ can be computed in $O(n32^n)$ time using $O(n^2)$ space. The other bounds are obtained similarly.
\end{proof}

The only other algorithms we know of for $Q(n)$ and $T(n)$ with nontrivial complexity bounds run in time $O(f(n)8^n)$ where $f(n)$ is a low-order polynomial \cite{alg}. However, these algorithms require $O(n^28^n)$ space, whereas we only require $O(n^2)$ space. We do not know of any algorithms with nontrivial complexity bounds for the other two problems.

%%%%%%%%%%%%%%%%%%%%%%%%%%%%%%%%%%%%%%%%%%%%%%%%%%%%%%%
\subsection{Extension: Latin Squares}

A Latin square of order $n$ is an arrangement of $n$ copies of the integers $1,2, \ldots,n$ in an $n$-by-$n$ grid such that every integer appears
exactly once in each row and column. We now show how an expression for $L_n$, the number of Latin squares of order $n$, follows naturally from the method used above.

\begin{lemma}
\label{easy}
Let $\mathbf{B}_n$ be the $n$-by-$n$ obstruction matrix containing the variables $(x_1, x_2, \ldots, x_{n^2})$ defined by $(\mathbf{B}_n)_{i,j} = x_{i + n(j-1)}$. Let $\mathbf{A}_n$ be the $n^2$-by-$n^2$ block diagonal matrix 

$$\mathbf{\mathbf{A}_n} = \begin{bmatrix} 
\mathbf{B}_n & 0 & \cdots & 0 \\ 0 & \mathbf{B}_{n} & \cdots & 0 \\
\vdots & \vdots & \ddots & \vdots \\
0 & 0 & \cdots & \mathbf{B}_{n}
\end{bmatrix}.$$
Then $L_n = g(\mathrm{per}(\mathbf{A}_n), n^2)$.
\end{lemma}

\begin{proof}
A Latin square of order $n$ can be thought of as an ordered set of $n$ disjoint permutation matrices of order $n$. On the other hand, a term in $\mathrm{per}(\mathbf{A}_n)$ can be thought of as an ordered set of $n$ permutation matrices of order $n$, one along each copy of $\mathbf{B}_n$. If this term contains $n^2$ distinct variables, these permutation matrices must be disjoint. Therefore the sum of the coefficients of the terms in $\mathrm{per}(\mathbf{A}_n)$ containing $n^2$ distinct variables is exactly $L_n$.
\end{proof}

\begin{theorem}
Let $L_n$ be the number of Latin squares of order $n$. Then 
$$L_n = \sum_{\mathbf{M} \in M_n} (-1)^{\sigma (\mathbf{M})+ n} \mathrm{per} (\mathbf{M})^n$$
where $M_n$ is the set of all $(0,1)$ $n$-by-$n$ matrices, and $\sigma (\mathbf{M})$ is the number of nonzero entries in $\mathbf{M}$.
\end{theorem}

\begin{proof}
From Lemma 3.8 and Theorem 3.5, it follows that 
\begin{align*}
L_n &= \sum_{i=1}^{n^2} (-1)^{i+n^2} f(\mathrm{per}(\mathbf{A}_n), i) \\
&= \sum_{i=1}^{n^2} (-1)^{i+n} \sum_{s \in S_{n^2,i}} \mathrm{per} (\mathbf{A}_n |s)\\
&= \sum_{i=1}^{n^2} (-1)^{i+n} \sum_{s \in S_{n^2,i}} \mathrm{per} (\mathbf{B}_n |s)^n,\\
\intertext{where the last step follows from the fact that $\mathrm{per}(\mathbf{A}_n) = \mathrm{per}(\mathbf{B}_n)^n$. Because $\mathbf{B}_n | u \neq \mathbf{B}_n | v$ if $u \neq v$, we can rewrite this as} 
L_n &= \sum_{\mathbf{M} \in M_n} (-1)^{\sigma (\mathbf{M})+ n} \mathrm{per} (\mathbf{M})^n. \tag*{\qedhere}
\end{align*}
\end{proof}
This formula was first given in \cite{latin}.

\section{Lower Bounds}
In the last section, we showed that sums of coefficients in the permanents of the obstruction matrices $\mathbf{Q}_n,\mathbf{T}_n,\mathbf{S}_n,$ and $\mathbf{Z}_n$ correspond to the values of $Q(n),T(n),S(n)$, and $TS(n)$, respectively. We then gave a closed-form expression for the function $g$ that computes these sums. More precisely, $g(P,k)$ was the sum of the coefficients of the terms in the polynomial $P$ containing $k$ distinct variables.

Now since each entry in $\mathbf{Q}_n$ is a monomial with coefficient $1$, the coefficient of a term in $\det(\mathbf{Q}_n)$ is at most the coefficient of the corresponding term in $\mathrm{per}(\mathbf{Q}_n)$. Therefore $|g(\det(\mathbf{Q}_n),2n)| \le g(\mathrm{per}(\mathbf{Q}_n),2n) = Q(n)$. The same argument applies to the other problems. As a result we have the following corollary:

\begin{corollary}
\label{lowerbounds}
Let $S_{m,k}$ be the subset of $\{0,1\}^m$ that consists of the tuples containing $k$ ones, $U_n$ the set of all $n$-by-$n$ $(0,1)$ diagonally constant matrices, and $V_n$ the set of all $n$-by-$n$ $(0,1)$ circulant matrices. Then the following inequalities hold:
\begin{align*}
Q_{\det}(n) & := \bigg | \sum_{i=1}^{2n} (-1)^{i} \binom{4n-i-2}{2n-i}\sum_{s \in S_{4n-2, i}} \det(\mathbf{Q}_n |s) \bigg | \le Q(n),\\
T_{\det}(n) & := \bigg | \sum_{i=1}^{2n} (-1)^{i+n} \sum_{s \in S_{2n, i}} \det(\mathbf{T}_n |s) \bigg | \le T(n), \\
S_{\det}(n) & := \bigg |\sum_{\mathbf{M} \in U_n} (-1)^{\gamma(\mathbf{M})} \det(\mathbf{M}) \binom{2n - \gamma(\mathbf{M}) -1}{n-\gamma(\mathbf{M})} \bigg | \le S(n),\\
TS_{\det}(n) & := \bigg |\sum_{\mathbf{M} \in V_n} (-1)^{\sigma(\mathbf{M})} \det(\mathbf{M}) \bigg | \le TS(n),
\end{align*}
where $\gamma(\mathbf{M})$ is the number of nonzero diagonals in $\mathbf{M}$, and $\sigma(\mathbf{M})$ is the number of ones in the first row of $\mathbf{M}$.
\end{corollary}

We now show that $TS_{\det}(n)$ is the permanent of Schur's matrix of order $n$; see \cite[Sequence A003112]{oeis}.

Let $\mathbf{M}_n = (\epsilon^{jk})$ be an $n$-by-$n$ matrix where $\epsilon$ is an $n$th root of unity, and let $P_n = \mathrm{per}(\mathbf{M}_n)$. The matrix $\mathbf{M}_n$ is known as Schur's matrix of order $n$. It has been of interest in number theory, statistics, and coding theory. Its permanent is the topic of \cite{schur}.

\begin{theorem}
For all $n$, $|P_n| \le TS(n) \le S(n)$.
\end{theorem}
\begin{proof}
From Corollary \ref{lowerbounds}, it suffices to show that $TS_{\det}(n) = |P_n|$. This follows immediately from the fact that $P_n = g(\det(\mathbf{Z}_n),n)$ \cite{schur}.
\end{proof}

%%%%%%%%%%%%%%%%%%%%%%%%%%%%%%%%%%%%%%%%%%%%%%%%%%%%%%%
% \bibliographystyle{plain} 
% \bibliography{myBibFile} 
% If you use BibTeX to create a bibliography
% then copy and past the contents of your .bbl file into your .tex file

\end{document}